\documentclass[letter,11pt]{article}
\usepackage{breakcites}

\usepackage[colorlinks]{hyperref}
 \hypersetup{linkcolor=blue,filecolor=blue,citecolor=blue,urlcolor=blue}
\usepackage{xspace}

\newcommand{\fullversion}[1]{#1}
\newcommand{\submversion}[1]{}

\fullversion{

\usepackage{amsmath}
\usepackage{amsthm}
\usepackage{cleveref}
\newtheorem{theorem}{Theorem}

\newtheorem{definition}[theorem]{Definition}
\newtheorem{lemma}[theorem]{Lemma}
\newtheorem{claim}[theorem]{Claim}
\newtheorem{remark}[theorem]{Remark}
\newtheorem{corollary}[theorem]{Corollary}

\crefname{claim}{claim}{claims}
\crefname{proposition}{proposition}{propositions}
\crefname{definition}{definition}{definitions}
\crefname{corollary}{corollary}{corollaries}
\crefname{conjecture}{conjecture}{conjectures}

\newcommand{\secparam}{\lambda}

}

\usepackage{caption}
\usepackage{subcaption}

\usepackage{xspace}
\fullversion{
\usepackage{fullpage}}

\usepackage{bbold}
\usepackage{verbatim}
\usepackage{mathtools}
\usepackage{colortbl}

\newcommand{\cp}{\mathsf{CP}}
\newcommand{\fclass}{\mathcal{F}}

\newcommand{\Sim}{\mathsf{Sim}}

\newcommand{\density}[1]{\mathcal{D}\left( \mathcal{H}_{#1} \right)}

\usepackage{pifont}

\usepackage{qcircuit}

\newcommand{\bra}[1]{\langle #1|}
\newcommand{\ket}[1]{|#1\rangle}

\newcommand{\ketbra}[2]{|#1\rangle\langle #2|}

\usepackage{amsfonts}
\usepackage[utf8]{inputenc}
\usepackage{tabularx}
\usepackage[normalem]{ulem}
\usepackage{setspace}
\usepackage{fancyhdr}
\usepackage{lastpage}
\usepackage{extramarks}
\usepackage{chngpage}
\usepackage{soul}
\usepackage{xspace}
\usepackage{bm}
\usepackage{nicefrac}
\usepackage{paralist}
\usepackage{enumerate}
\usepackage[shortlabels]{enumitem}
\usepackage{enumitem}
\usepackage{tikz}

\usepackage[justification=centering]{caption}
\newcommand{\negl}{\mathsf{negl}}

\newcommand{\ignore}[1]{}

\submversion{

}
\fullversion{

}

\newcommand{\cktclass}{\mathcal{C}}

\newcommand{\poly}{\mathrm{poly}}

\newcommand{\distr}{\mathcal{D}}

\newcommand{\gen}{\mathsf{Gen}}

\newcommand{\adversary}{\mathcal{A}}  

\DeclareMathOperator*{\prob}{\mathsf{Pr}} 

\newcommand{\setup}{\mathsf{Setup}}

\newcommand{\enc}{\mathsf{Enc}}
\newcommand{\dec}{\mathsf{Dec}}

\newcommand{\pk}{\mathsf{PK}}
\newcommand{\ct}{\mathsf{CT}}

\newcommand{\qfhe}{\mathsf{QFHE}}

\newcommand{\eval}{\mathsf{Eval}}

\newcommand{\sk}{\mathsf{SK}}

\newcommand{\lobf}{\mathsf{Obf}}
\newcommand{\leval}{\mathsf{ObfEval}}
\newcommand{\lockC}{\mathbf{C}}

\newcommand{\ext}{\mathsf{Ext}}

\newcommand{\tr}{\mathsf{Tr}}
\newcommand{\cE}{\mathcal{E}}
\newcommand{\cD}{\mathcal{D}}

\newcommand{\cH}{\mathcal{H}}
\newcommand{\cM}{\mathcal{M}}
\newcommand{\cC}{\mathcal{C}}

\newcommand{\cY}{\mathcal{Y}}
\newcommand{\cA}{\mathcal{A}}

\newcommand{\uniform}{\xleftarrow{\$}}
\newcommand{\given}{\mid}

\newcommand{\alice}{\mathcal{A}}
\newcommand{\bob}{\mathcal{B}}
\newcommand{\charlie}{\mathcal{C}}

\newcommand{\abc}{(\alice,\bob,\charlie)}

\newcommand{\id}{\mathbf{id}}
\newcommand{\Fatih}[1]{{\color{blue} F: #1}}
\renewcommand{\Fatih}[1]{}

\DeclareMathOperator*{\E}{\mathbb{E}}

\newcommand{\trD}[2]{T(#1, #2)}

\newcommand{\brackets}[1]{\left( #1 \right)}
\newcommand{\bracketsSquare}[1]{\left[ #1 \right]}
\newcommand{\bracketsCurly}[1]{\left\{ #1 \right\}}

\newcommand{\bracketsC}{\bracketsCurly}

\newcommand{\carom}{\text{CAROM}}

\newcommand{\pr}[1]{\prob \left[ #1 \right]}

\newcommand{\cpprime}{\widetilde{\cp}}
\newcommand{\evalprime}{\widetilde{\eval}}

\newcommand{\from}{\leftarrow}

\newcommand{\cO}{\mathcal{O}}

\ignore{

}

\newcommand{\cnc}{\mathsf{cnc}}

\newcommand{\simulator}{\mathsf{Sim}}
\renewcommand{\secparam}{\lambda}

\usepackage{breakcites}

\newcommand{\distrc}{\mathcal{D}_\cC}

\fullversion{
\title{A Note on Copy-Protection from Random Oracles}
\author{Prabhanjan Ananth\thanks{prabhanjan@cs.ucsb.edu}\\UCSB \and Fatih Kaleoglu\thanks{kaleoglu@ucsb.edu}\\UCSB}

\pagenumbering{arabic}
}

\fullversion{

}
\submversion{

}
\date{}

\begin{document}
\maketitle
\begin{abstract}
\noindent Quantum copy-protection, introduced by Aaronson (CCC'09), uses the no-cloning principle of quantum mechanics to protect software from being illegally distributed. Constructing copy-protection has been an important problem in quantum cryptography. 
\par Since copy-protection is shown to be impossible to achieve in the plain model, we investigate the question of constructing  copy-protection for arbitrary classes of unlearnable functions in the random oracle model. We present an impossibility result that rules out a class of copy-protection schemes in the random oracle model assuming the existence of quantum fully homomorphic encryption and quantum hardness of learning with errors. En route, we prove the  impossibility of approximately correct copy-protection in the plain model. 

\end{abstract}

\section{Introduction}

\label{sec:intro}
\noindent Quantum copy-protection, introduced by Aaronson~\cite{Aar09}, is a foundational concept in quantum cryptography. It stipulates that the no-cloning principle of quantum mechanics~\cite{Dieks82,WZ82} can be employed to protect against illegal distribution of software. In more detail, an efficient adversary, on input a copy-protected software (represented as a quantum state), cannot create two copies of software, possibly entangled with each other, such that both copies compute the same functionality as the original software. 
\par The primitive of copy-protection can be classified under the broad area of unclonable cryptography, which deals with using the no-cloning principle to design cryptographic primitives with security properties that are classically unachievable. Many interesting primitives in this category, such as quantum money~\cite{wiesner83,AC12,Zhandry19,RS19}, one-shot signatures~\cite{AGKZ20}, single-decryptor encryption \cite{GZ20,CLLZ21}, unclonable encryption~\cite{Got02,BL19}, and encryption with certifiable deletion~\cite{BI20}, can be seen as copy-protecting specific functionalities. 
\newcommand{\calO}{\mathcal{O}}
\par The focus of our work is on understanding the feasibility of constructing copy-protection for {\em all} classes of unlearnable functions. Ananth and La Placa \cite{ALP20} show that there are functions that cannot be copy-protected in the plain model. Thus, one has to rely on alternate models to construct copy-protection. 
\par In this work, we restrict our attention to the random oracle model. Interestingly, random oracles have been helpful for achieving copy-protection for specific classes of functions. Coladangelo, Majenz and Poremba~\cite{CMP20} showed the existence of  copy-protection for multi-bit output functions in the random oracle model. Ananth, Kaleoglu, Liu, Li and Zhandry~\cite{AKLLZ22} presented a new construction of copy-protection for single-bit output point functions also in the random oracle model.  However, the existence of  copy-protection for {\em all} classes of unlearnable functions in the random oracle model is still yet to be explored. 
\par Before we delve into this direction further, we first need to model the type of access the algorithms in the copy-protection scheme and the adversarial entities will have, with the random oracle. There are two types of accesses we can consider. The first one is classical access, where the interface is entirely classical: the algorithms submit a binary string $x$ to the oracle and get back $f(x)$, where $f$ is a classical random function implemented by the oracle. We will refer to this setting as \emph{classical-accessible random oracle model} ({\carom}). The second type is quantum access: the algorithms submit a query of the form $\sum_{x \in \{0,1\}^n} \alpha_x \ket{x}\ket{y_x}$ and get back $\sum_{x \in \{0,1\}^n} \alpha_x \ket{x}\ket{f(x) \oplus y_x}$. This setting was defined by Boneh, Dagdelen, Fischlin, Lehmann, Schaffner and Zhandry~\cite{BDFLSZ11} as the \emph{quantum random oracle model} (QROM).  
\par As argued in \cite{BDFLSZ11}, QROM is the preferred model over {\carom} under most circumstances. For instance, if the adversary has access to any classical code that computes the function $f$, then it can always run this code coherently, thus achieving quantum access. Nonetheless, as per the impossibility result of \cite{ALP20}, one cannot simply instantiate a generic copy-protection scheme in QROM by replacing the random function $f$ with a concrete function, e.g. a heuristically secure hash function such as SHA-512 or an obfuscated pseudorandom function. In contrast, classical-access is appropriate for alternative methods of instantiation, such as using trusted hardware or a trusted party, where classical interface can be enforced. Therefore, we argue that achieving copy-protection in {\carom} is still meaingful.

\subsection{Our Result}
We make progress towards understanding the feasibility of copy-protection in the random oracle model. It turns out that proving impossibility of copy-protection in the QROM seems quite challenging and thus, we focus on the CAROM setting. Note that QROM and CAROM are incomparable models, since QROM gives more power to both the honest algorithms and the adversary.
\par We show that copy-protection of arbitrary unlearnable functions is impossible in {\carom}. This implies that if copy-protection is possible with a random oracle, then honest algorithms must query the oracle in superposition.
\noindent 

\begin{theorem}
Assuming unbounded fully homomorphic encryption for quantum computations~\cite{mahadev2018classical,brakerski2018quantum} and quantum hardness of learning with errors (QLWE), there exists a class of unlearnable functions ${\cal F}$ such that quantum copy-protection for ${\cal F}$ is impossible in \carom.
\end{theorem}

\noindent The above result suggests that if one were to base copy-protection in the classical-accessible oracle models then the oracle needs to have some structure. We note that in both the works of~\cite{CMP20} and~\cite{AKLLZ22}, the copy-protection algorithm only makes classical queries to the oracle whereas the evaluation algorithm makes quantum queries.
\par At the heart of the above theorem is a new impossibility result for copy-protection in the plain model. Specifically, we show that even copy-protection with approximate correctness can be ruled out in the plain model (\Cref{thm:imp:approxcp}), while prior works~\cite{ALP20,ABDS21} only ruled out copy-protection with statistical correctness.

\par We combine this impossibility result with a generic transformation from a copy-protection scheme in {\carom} to an approximate copy-protection scheme in the plain model (\Cref{lem:romtoplain}) to rule out copy-protection of arbitrary unlearnable functions in {\carom} (\Cref{thm:imp_main}). Our transformation crucially relies on recording the oracle queries of the honest algorithms and efficient simulation of a random oracle consistent with a prerecorded database. Recording quantum queries to a random oracle can be achieved using Zhandry's compressed oracle technique \cite{Zha19_compressed}, but to our knowledge it is not known how to efficiently simulate a quantum random oracle consistent with a prerecorded database. Therefore, any construction in QROM where honest parties make superposition queries to the oracle would circumvent our impossibility result.

\subsection{Related Work}
\label{sec:relatedwork} Aaronson~\cite{Aar09} was the first to study copy-protection in the oracle models. Their construction relied upon oracles implementing quantum functionalities. Recently, Aaronson, Liu, Liu, Zhandry and Zhang~\cite{ALLZZ20} constructed copy-protection in a quantum-accessible oracle model. As mentioned earlier, two works~\cite{CMP20,AKLLZ22} present constructions of copy-protection for point functions in the quantum random oracle model. 
\par In a recent exciting work, Coladangelo, Liu, Liu and Zhandry~\cite{CLLZ21} proposed the first construction of copy-protection for a restricted class of functions, namely pseudorandom functions, in the plain model, based on post-quantum indistinguishability obfuscation and post-quantum one-way functions. Another work by Ananth and Kaleoglu~\cite{AK21} present a construction of approximately correct copy-protection for a class of point functions from post-quantum one-way functions. Some recent works~\cite{ALP20,ALLZZ20,KNY20,broadbent2021secure} present constructions of weaker notions of copy-protection. 

\section{Overview of Techniques}
\label{sec:overview:impossibility}
\noindent In this section, we explain the main ideas behind our result. Inspired by the ideas developed in the program obfuscation literature~\cite{BV16,CKP15}, we design a two step approach to proving this impossibility result. 
\begin{itemize}
    \item {\em Ruling out approximate copy-protection in the plain model}: In the first step, we rule out a notion of copy-protection where correctness is only guaranteed for a large constant fraction of inputs.
    \item {\em From CP using oracles to approximate CP in the plain model}: In the second step, we show that copy-protection in the classical-accessible random oracle model implies approximate copy-protection in the plain model. In other words, given any copy-protection in the classical-accessible random oracle model, we can generically get rid of the random oracle model at the cost of weakening the correctness guarantee. 
\end{itemize}

\paragraph{Ruling out Approximate Copy-Protection.} We show that copy-protection, where correctness is guaranteed for a fraction of inputs, say $(1- \varepsilon)$, is impossible to achieve in the plain model. This strengthens the previous result~\cite{ALP20} which ruled out copy-protection with correctness negligibly close to 1. 
\par Our goal is to transform this scheme, call it $CP$, such that in the transformed scheme, call it $CP'$, on every input, the evaluation of the copy-protected state is correct with probability $(1- \varepsilon)$ (i.e., per-input $(1- \varepsilon)$-correctness); assume for now, that $C$ is a boolean circuit with 1-bit output. The scheme $CP'$ is designed as follows: to copy-protect $C$, compute a copy-protection of a circuit $G$, with respect to the scheme $CP$, such that $G$ takes as input an encryption of $x$, homomorphically computes $C$ on encryption of $x$ and outputs the result. The output is copy-protection of $G$ along with the public key-secret key pair of the encryption scheme. During the evaluation process of $CP'$, first encrypt the input $x$, run the $CP$ copy-protection of $G$ to obtain encryption of $C(x)$ and finally, decrypt the answer to obtain $C(x)$ in the clear.
\par This idea was notably developed by Bitansky and Vaikuntanathan~\cite{BV16} in the context of indistinguishability obfuscation. To make this idea work, we would need fully homomorphic encryption schemes for quantum computations (QFHE). Moreover, we require the QFHE scheme to satisfy circuit privacy; that is, the homomorphically evaluated ciphertext does not leak information about the circuit being used during evaluation. This property is necessary because the evaluator who obtains a classical description of the function can trivially break the copy-protection security.
Fortunately, a recent work by Chardouvelis, D\"ottling, and Malavolta~\cite{CDM20} demonstrates the existence of a QFHE scheme satisfying the circuit privacy property we need\footnote{We note that the definition considered in~\cite{CDM20} is weak: for example, they do not even handle adversaries who can entangle the quantum messages with their private state. We provide a stronger definition in this paper and remark that the construction provided in~\cite{CDM20} already satisfies this stronger definition.}. \\

\noindent {\em Failure of Majority Argument.} Suppose we manage to reduce the feasibility of copy-protection with correctness over a $(1- \varepsilon)$-fraction of inputs to per-input $(1- \varepsilon)$-correctness. The next step would be to rule out copy-protection with per-input correctness. A natural attempt would be to consider multiple copies of copy-protection and take a majority vote; this would improve the correctness to be close to 1. While this argument would work for program obfuscation, this unfortunately fails in the context of copy-protection. The reason is simple: once you have many copies of copy-protection, an adversary can now distribute each copy to a different individual, thus breaking the security of copy-protection.
\par Thus, we need to figure out an alternate method to rule out quantum copy-protection with per-input correctness guarantees. As done in prior works, we rely upon non-black box techniques to rule out copy-protection. We need to be especially careful when invoking non-black box arguments in the per-input $(1- \varepsilon)$-correctness setting, as every evaluation of the copy-protected state significantly degrades the correctness guarantee of the original state. If $\varepsilon$ is the correctness error then after $k$ evaluations, the trace distance between the original state and the new state is $k\sqrt{ \varepsilon}$ (by quantum union bound). If $k \geq \frac{1}{\sqrt{\varepsilon}}$ then the new state is useless. Luckily, the quantum non-black box technique of~\cite{ALP20} involves the attacker only making two evaluations of the copy-protected state. 

In the technical sections, we combine the self-reducibility technique (presented above) with the non-black box technique of~\cite{ALP20} to rule out copy-protection for $(1- \varepsilon)$-fraction of inputs. 

\paragraph{From CP using oracles to approximate CP.} The next step would be to rule out copy-protection in the classical-accessible random oracle model. We rely upon ``de-oracle-izing" techniques developed in the context of obfuscation \cite{CKP15}; the idea is to remove the use of random oracle in the construction at the cost of weakening the correctness guarantee. While the overall template is inspired from~\cite{CKP15}, the actual construction is different. 
\par In more detail, given a copy-protection scheme $(\cp, \eval)$ in the classical-accessible random oracle model, where $\cp$ represents the copy-protection algorithm and $\eval$ represents the evaluation algorithm, define a copy-protection scheme $(\cpprime, \evalprime)$ in the plain model as follows:
\begin{enumerate}
    \item Using a random oracle $\mathcal{O}$ simulated \textit{on-the-fly}, $\cpprime(f)$ runs $\cp^\mathcal{O}(f)$ to obtain a copy-protected program $\rho_f$. It then runs a random (polynomial) number of test executions $\eval^\mathcal{O}(\rho_f, x_i)$ for randomly chosen inputs $x_i$. It records all the queries made by $\eval$ during the test executions in a database $D$. Finally, it samples a set of random oracle answers $R$ and outputs $(\rho, D, R)$.
    
    \item $\evalprime((\rho_f, D, R), x)$ simulates a random oracle $\mathcal{O'}$ using the database $D$ and using $R$ for queries not recorded in $D$. It runs $\eval^\mathcal{O'}(\rho_f, x)$ and outputs the answer. 
\end{enumerate}

\noindent The key difference between our construction and the construction of \cite{CKP15} is that we choose the number of test executions at random. In the classical world, running more test executions can never hurt the simulation. In the quantum world, on the other hand, every test execution can significantly alter the state $\rho_f$ by performing measurements. This could be true even if $\rho_f$ is reusable, in the sense that its correctness guarantee is preserved after polynomially many evaluations\footnote{These \emph{forced} measurements are a unique feature of the classical-accessible oracle setting, since quantum queries need not be measured at the time of the query by the deferred measurement principle.}. For instance, the state $\rho_f$ could maintain a counter which affects the oracle queries made by $\eval$. To ensure that the oracle queries $\evalprime$ needs to answer are captured in the database $D$ with probability close to 1, even if $\rho_f$ keeps changing, we choose the number of test queries at random and sufficiently large. 
\par Perhaps surprisingly, our technique also  improves the classical impossibility of obfuscation in the random oracle model! Specifically, it reduces the number of test executions needed by a factor of $N$ when compared to \cite{CKP15}, where $N$ is a query-bound for $\eval$. 

\subsection{Acknowledgements}
We thank Luowen Qian for valuable discussions.

\subsection{Organization} \label{sec:org}
Preliminaries are described in~\Cref{sec:prelims}. The impossibility of approximate copy-protection in the plain model is presented in~\Cref{sec:impossibilityplain}. We prove the impossibility of copy-protection in CAROM in~\Cref{sec:imp:carom}.  
\newcommand{\pirateexp}[4]{\mathsf{PirExp}^{#1,#2}_{#3,#4}}
\newcommand{\distrclass}{\mathfrak{D}_X}
\newcommand{\puzzlegame}{ \mathsf{SOLVE2}_{\mathcal{T}, \mathcal{Z}} }
\newcommand{\genO}{\mathsf{GenOracle}}

\section{Preliminaries}
\label{sec:prelims}

\subsection{Notation}
\label{sec:prelims_notation}
We denote by $x \uniform X$ the sampling of an element $x$ from the uniform distribution over $X$. We denote by $\secparam$ the security parameter. We use the terms \textit{function} and \textit{circuit} interchangeably. We denote by $\negl(\cdot)$ a generic negligible function. $[N]:= \{1,2,\dots,N\}$. If an algorithm $\alice$ has oracle access to $\mathcal{O}$, we may write $\alice^{\mathcal{O}}$ to emphasize this fact. We call a function $\epsilon: \mathbb{Z}^+ \to \mathbb{R}^+ \cup \{0\}$ \textit{noticable} if there exists a polynomial $p$ such that $\epsilon(n) > \frac{1}{p(n)}$ for sufficiently large $n$. We say \emph{overwhelming} probability to mean probability $1 - \negl(\secparam)$. We sometimes shorten a string of zeros to 0, when the length will be understood from the context. We write $\alice(x;r)$ when a classical algorithm $\alice$ is run on input $x$ using specified randomness $r$.

\subsection{Quantum Computing Basics}
\label{sec:app:prelims_quantumcomputing}
For a finite set (register) $X$ we denote by $\cH_X$ the Hilbert Space generated by the basis $\left\{ \ket{x} \right\}_{x \in X}$. We denote by $\cD(\cH)$ the set of valid mixed quantum states over Hilbert space $\cH$, i.e. linear, positive-semidefinite operators over $\cH$ with unit trace, also known as density operators. We define a quantum polynomial time (QPT) algorithm as a family of generalized quantum circuits $\bracketsC{\alice_\secparam}_{\secparam \in \mathbb{N}}$ such that each $\alice_\secparam$ contains at most $p(\secparam)$ input/output qubits (including auxiliary qubits) and at most $p(\secparam)$ gates from a universal gate set (such as $\bracketsC{CNOT, H, T}$), for some polynomial $p(\cdot)$.

\subsubsection{Trace Distance}
\label{sec:trd}
A common way to measure dissimilarity between two quantum states $\rho, \sigma \in \cD(\cH)$ is the (normalized) trace distance, defined as $$\trD{\rho}{\sigma} := \frac{1}{2}|\rho - \sigma |_{tr} = \frac{1}{2}\tr\left(\sqrt{(\rho - \sigma)^2}\right).$$
It is a fact that no quantum algorithm can increase the trace distance between two quantum states, i.e. $\trD{\alice(\rho)}{\alice(\sigma)} \le \trD{\rho}{\sigma}$ for any $\alice, \rho, \sigma$. When two quantum states $\rho(\secparam), \sigma(\secparam)$ depend on a security parameter and satisfy $\trD{\rho}{\sigma} \le \negl(\secparam)$, we write $\rho \approx_s \sigma$, meaning $\rho$ and $\sigma$ are \textit{statistically close}.

\subsubsection{Measurement}
\label{sec:measurement}
A quantum measurement on a Hilbert Space $\cH$ with a finite set of outcomes $[m]$ can be represented by a set of operators $(M_i)_{i \in [m]}$ satisfying $\sum_{i \in [m]} M_i^\dagger M_i = I$. Given a quantum state $\rho \in \cD(\cH)$, the probability of outcome $i$ is given by $\tr(M_i\rho M_i^\dagger)$, whereas the post-measurement state after measuring outcome $i$ is given by $\frac{M_i \rho M_i^\dagger}{\tr(M_i\rho M_i^\dagger)}$. If the measurement outcome is not revealed, then the post-measurement state equals the mixture $\sum_{i \in [m]} M_i\rho M_i^\dagger$.

\subsubsection{Almost As Good As New Lemma}
\label{sec:aagan}
The following lemma from \cite{aaronson2004limitations} is widely used in literature\footnote{In quantum information theory, it is known as \textit{the gentle measurement lemma} \cite{Winter_1999}.}.
\begin{lemma}[Almost as Good as New Lemma]
\label{lem:aagan}
Let $\rho$ be a quantum state and $(M_0, M_1)$ be a $2$-outcome measurement such that $\tr(M_0\rho M_0^\dagger) \ge 1 - \varepsilon$. Then, after this measurement is performed on $\rho$, it is possible to recover a state $\rho'$ such that $\trD{\rho}{\rho'} \le \sqrt{\varepsilon}$. In addition, the recovery procedure is independent of $\rho$, and is efficient (runs in polynomial time) given that $(M_0,M_1)$ is efficient.
\end{lemma}

\begin{corollary}
\label{cor:aagan}
Let $\alice$ be a quantum algorithm that given as input a quantum state $\rho$ and a classical string $x$, outputs a classical string $y$. Suppose that \begin{align}
\label{eq:correctness}    \pr{f(x) \leftarrow \alice(\rho, x)} \ge 1 - \varepsilon(x),
\end{align}
    where $\rho$ is a quantum input state, $x$ is a classical input, and $f(x)$ is a classical deterministic function. Then, there exists a quantum algorithm $\alice'$ such that (1) $\alice'(\rho, x)$ outputs a classical string $y$ identically distributed to output of $\alice(\rho, x)$, and (2) in addition $\alice'$ outputs a residual state $\rho'$ satisfying $\trD{\rho}{\rho'} \le \sqrt{\varepsilon(x)}$. Moreover, $\alice'$ is efficient given that $\alice$ is efficient.
\end{corollary}

\begin{proof}[Sketch]
By deferred measurement principle, we can modify $\alice$ so that it applies a unitary transformation followed by measuring and outputting a value $y$. Fix $x \in X$, then the statement follows after applying \Cref{lem:aagan} with respect to the 2-outcome measurement which checks whether $y = f(x)$ or not.
\end{proof}

\subsection{Quantum Fully Homomorphic Encryption}
\label{sec:prelims_qfhe}
A fully homomorphic encryption scheme allows for publicly evaluating an encryption of $x$ using a function $f$ to obtain an encryption of $f(x)$. Traditionally $f$ has been modeled as classical circuits but in this work, we consider the setting when $f$ is modeled as quantum circuits and when the messages are quantum states. This notion is referred to as quantum fully homomorphic encryption (QFHE). We state our definition borrowed directly from\footnote{With the slight modification that the evaluation key is included as part of the public key.} \cite{broadbent2015quantum}.

\begin{definition} Let $\cM$ be the Hilbert space associated with the message space (plaintexts) and $\cC$ be the Hilbert space associated with the ciphertexts. A quantum fully homomorphic encryption (QFHE) scheme is a tuple of QPT algorithms  $\qfhe=(\gen,\enc,\dec,\allowbreak \eval)$:
\begin{itemize}
    \item $\qfhe.\gen(1^\secparam)$: Takes as input a security parameter $\secparam$ in unary; outputs a classical public-secret key pair, $(\pk,\sk)$.
    \item $\qfhe.\enc(\pk,\cdot):\cD(\cM)\rightarrow \cD(\cC)$: Takes as input a public key $\pk$ and a quantum message $\rho$; outputs a quantum ciphertext $\sigma$.
    \item $\qfhe.\dec(\sk,\cdot):\cD(\cC)\rightarrow \cD(\cM)$: Takes as input a secret key $\sk$ and a quantum ciphertext $\sigma$; outputs a quantum message $\rho$.
    \item $\qfhe.\eval(\pk, \cE, \cdot ):\cD(\cC^{\otimes n})\rightarrow \cD(\cC^{\otimes m})$: Takes as input a public key $\pk$, description of a quantum circuit $\cE: \cD(\cM^{\otimes n}) \rightarrow \cD(\cM^{\otimes m})$, and a tuple of quantum ciphertexts $\sigma \in (\cD(\cC^{\otimes n}))$; outputs a tuple of quantum ciphertexts $\sigma' \in \cD(\cC^{\otimes m})$.

\end{itemize}
\end{definition}
\noindent Semantic security and compactness are defined analogously to the classical setting, and we defer to~\cite{broadbent2015quantum} for a definition.
\noindent For the impossibility result, we require a $\qfhe$ scheme where ciphertexts of classical plaintexts are also classical. Given any classical message $x \in \{0,1\}^k$, we want $\qfhe.\enc(\pk, \ket{x}\bra{x})$ to be a computational basis state $\ket{z}\bra{z}$ for some $z \in \{0,1\}^l$ (here, $l$ is the length of ciphertexts for $k$-qubit messages). In this case, we write $\qfhe.\enc_\pk(x)$. We also want the same to be true for evaluated ciphertexts, i.e. if $\cE(\ket{x}\bra{x})=\ket{y}\bra{y}$ for some $x \in \{0,1\}^n$ and $y \in \{0,1\}^m$, then 
$$\qfhe.\enc_\pk(y) \leftarrow \qfhe.\eval(\pk, \cE, \qfhe.\enc_\pk(x)) $$
is a classical ciphertext of $y$.

\subsubsection{Circuit Privacy}
\label{sec:prelims_qfhe_circuitprivacy}
An additional property we need a QFHE scheme to satisfy is called \textit{malicious circuit privacy}. Informally, it states that the evaluation algorithm does not leak any information about the circuit being evaluated, i.e. given a homomorphic evaluation of a circuit $\cE$, an adversary cannot learn any information she would not have learned given only black-box access to $\cE$. We give the formal definition below, adapted from\footnote{The difference is that we sensibly strengthen the definition to account for adversaries with a possibly entangled auxiliary register $B$. Even though this was not formally considered by \cite{CDM20}, a quick analysis of their argument reveals that their construction indeed satisfies this stronger definition of circuit privacy.} \cite{CDM20}:
\begin{definition}[Malicious Circuit Privacy]
\label{def:qfhe_circuitprivacy}
A QFHE scheme $\qfhe = (\qfhe.\gen, \allowbreak \qfhe.\enc, \allowbreak \qfhe.\dec, \allowbreak \qfhe.\eval)$ satisfies \textbf{malicious circuit privacy} if there exist unbounded quantum algorithms $\qfhe.\ext$ and $\qfhe.\Sim$ such that for any (possibly invalid) public key $\pk^*$, any quantum circuit $\cE: \cD(\cM^{\otimes n}) \rightarrow \cD(\cM^{\otimes m})$, any ancillary register $A$ with $\cH_A = \cM^{\otimes \poly(\secparam)}$, and any $\rho\in \cD(\cC^{\otimes n} \otimes \cH_A)$ we have $$ \brackets{ \qfhe.\eval(\pk^*, \cE, \cdot) \otimes \id_A} \rho \approx_s \qfhe.\Sim(1^\secparam, (\cE \otimes \id_{AB})(\widetilde{\rho})),  $$
where $\widetilde{\rho} \in \cD(\cC^{\otimes n} \otimes \cH_A \otimes \cH_B)$ is obtained by applying $\qfhe.\ext(1^\secparam, \pk^*, \cdot)$ to the $\cC^{\otimes n}$ (ciphertext) register of $\rho$, and $B$ is an ancillary register.
\end{definition}
\par The definition can be alternatively stated as follows: for an unbounded adversary, one query $\qfhe.\eval$ evaluation of the circuit $\cE$ is no more powerful than one query to $\cE$ itself. Note also that $A$ is a hidden register held by an adversary that could be entangled with the ciphertext, so it is not accessible by the algorithms $\qfhe.\eval$ and $\qfhe.\ext$.

\paragraph{Instantiation.} Malicious Circuit Privacy can be achieved using a regular QFHE scheme \cite{CDM20}. We state their result as a theorem: \begin{theorem}
Assuming QLWE and circular security\footnote{This assumption is needed to evaluate circuits with unbounded depth in both classical and quantum FHE.}, there exists a QFHE scheme with malicious circuit privacy, which reduces to a classical FHE scheme for classical inputs.
\end{theorem}

\subsection{Compute-and-Compare Circuits and Quantum-Secure Lockable Obfuscation}
\label{sec:prelims_lobf}
\paragraph{Compute-and-compare Circuits.} The subclass of circuits that we are interested in is called compute-and-compare circuits, denoted by $\cktclass_{\cnc}$. A compute-and-compare circuit is of the following form: $\lockC[C,\alpha, \beta]$, where $\alpha$ is called a lock and $C$ has output length $|\alpha|$, is defined as follows: 
$$ \lockC[C,\alpha](x) =  \left\{\substack{\beta, \ \text{ if } C(x)=\alpha,\\ \ \\ 0,\ \text { otherwise }  } \right. $$   

\begin{definition}[Quantum-Secure Lockable Obfuscation]
An obfuscation scheme $(\lobf,\leval)$ for a class of compute-and-compare circuits $\cktclass_{\cnc}$ is said to be a \textbf{quantum-secure lockable obfuscation scheme} if the following properties are satisfied: 
\begin{itemize}
    \item It satisfies the functionality of obfuscation. 
    \item {\bf Security}: For every polynomial-sized circuit $C$, string $\beta \in \{0,1\}^{\poly(\secparam)}$,for every QPT adversary $\adversary$ there exists a QPT simulator $\simulator$ such that the following holds: sample $\alpha \xleftarrow{\$}  \{0,1\}^{\poly(\secparam)}$,
    $$\left\{ \lobf \left( 1^{\secparam},\lockC \right) \right\} \approx_{Q,\varepsilon} \left\{\simulator\left(1^{\secparam},1^{|C|} \right) \right\},$$
    where $\lockC$ is a circuit parameterized by $C,\alpha,\beta$ with $\varepsilon \leq \frac{1}{2^{|\alpha|}}$.  
\end{itemize}

\end{definition}

\paragraph{Instantiation.} The works of~\cite{WZ17,GKW17,GKVW19} construct a lockable obfuscation scheme based on polynomial-security of learning with errors. Since learning with errors is conjectured to be hard against QPT algorithms, the obfuscation schemes of~\cite{WZ17,GKW17,GKVW19} are also secure against QPT algorithms. Thus, we have the following theorem. 

\begin{theorem}[\cite{GKW17,WZ17,GKVW19}]
Assuming quantum hardness of learning with errors, there exists a quantum-secure lockable obfuscation scheme. 
\end{theorem}

\subsection{Classical-Accessible Oracles}
We call an oracle $\mathcal{O} : \mathcal{X} \to \mathcal{Y}$ \textit{classical-accessible} if it only accepts a classical query $ x \in X$, to which it responds with the classical value $\mathcal{O}(x)$. More formally, a query to $\cO$ made by a quantum algorithm $\alice$ can be described as the following quantum operation: measure the $\mathcal{X}$ (query) register of $\alice$ in the computational basis to obtain $x \in \mathcal{X}$, then XOR the value $\cO(x)$ to the $\cY$ (answer) register. In other words, if $\alice$ queries the oracle in state $\sum_{x,y,z}\alpha_{x,y,z} \ket{x}\ket{y}\ket{z}$, where $\mathcal{Z}$ is an ancillary register corresponding to the internal state of $\alice$, then the state after the query equals $\sum_{y,z} \beta_{y,z} \ket{x}\ket{y \oplus \cO(x)}\ket{z}$ with probability $p_x = \sum_{y,z} |\alpha_{x,y,z}|^2$, where $\beta_{y,z} = \alpha_{x,y,z}/\sqrt{p_x}$.

\subsubsection{Classical-Accessible Random Oracle Model} \label{sec:prelims_crom} In classical-accessible random oracle model (\carom{}), the function in question is assumed to be a uniformly random function $\mathcal{O} : \mathcal{X} \to \mathcal{Y}$ and modeled as a classical-accessible oracle. \\

\noindent A classical-accessible random oracle $\mathcal{O}: X \to Y$ can be efficiently simulated \textit{on-the-fly} as follows: \begin{itemize}
    \item Create an empty database $D \subset X \times Y$.
    \item On query $x \in X$, check if $D$ contains $x$. If yes, answer consistently; otherwise sample $y \uniform Y$, add $(x,y)$ to $D$, and answer with $y$.
\end{itemize}
\noindent A classical-accessible oracle simulated \textit{on-the-fly} is perfectly indistinguishable from a classical-accessible random oracle.

\subsection{Copy-Protection} \label{sec:prelims_cp}
Below we present the definition of a copy-protection scheme, adapted from \cite{broadbent2021secure} and originally due to \cite{Aar09}.
\begin{definition} [Copy-Protection Scheme] \label{def:copyprotection}
 Let $\fclass = \fclass(\secparam)$ be a class of efficiently computable functions of the form $f: X \to Y$. A copy protection scheme for $\fclass$ is a pair of QPT algorithms $(\cp, \eval)$ such that for some output space $\density{Z}$: \begin{itemize}
    \item \textbf{Copy Protected State Generation:} $\cp(1^\secparam, d_f)$ takes as input the security parameter $1^\secparam$ and a classical description $d_f$ of a function $f \in \fclass$ (that efficiently computes $f$). It outputs a mixed state $\rho_f \in \density{Z}$.
    \item \textbf{Evaluation:} $\eval(1^\secparam, \rho, x)$ takes as input the security parameter $1^\secparam$, a mixed state $\rho\in \density{Z}$, and an input value $x\in X$. It outputs a bipartite state $\rho' \otimes \ket{y}\bra{y} \in \density{Z} \otimes \density{Y}$.
    
\end{itemize}
\end{definition}
\noindent We will sometimes abuse the notation and write $\eval(1^\secparam, \rho, x)$ to denote either the classical output $y \in Y$ or the residual state $\rho'$ alone when the other is insignificant. Note that when we work in {\carom}, the algorithms $(\cp, \eval)$ will have access to the random oracle and be written as $(\cp^\cO, \eval^\cO)$. \\

\par There are three properties we require of a copy-protection scheme: correctness, reusability, and security.

\paragraph{Correctness:} Informally speaking, if an honestly generated copy-protected state $\rho_f$ for a function $f \in \fclass$ is honestly evaluated using $\eval$ on any input $x \in X$, the output should be $f(x)$. We present the formal definition below:

\begin{definition}[Correctness]
A copy-protection scheme $(\cp,\eval)$ for $\fclass$ is \textbf{$\delta$-correct} if the following holds for every $x \in X$, $f \in \fclass$:
$$\prob \left[ f(x) \leftarrow \eval(1^\secparam, \rho_f,x)\ :\ \rho_f \leftarrow \cp(1^{\secparam},d_f) \right] \geq \delta.$$
\end{definition}

\noindent Correctness property can be relaxed by averaging over inputs $x$ sampled from a distribution $\distr$.

\begin{definition}[Mean-Correctness] Let $\distr_X$ be a distribution over $X$.
A copy-protection scheme $(\cp,\eval)$ for $\fclass$ is\textbf{ $(\distr_X,\delta)$-mean-correct} if the following holds for every $f \in \fclass$:
$$\prob \left[ f(x) \leftarrow \eval(1^\secparam, \rho_f,x) \ :\ \substack{\rho_f \leftarrow \cp(1^{\secparam},d_f)\\ \ {x \leftarrow \distr_X}} \right] \geq \delta. $$
\end{definition}

\paragraph{Reusability:} The correctness notions we define above are for a single evaluation. In practice, we would like our copy-protected program to evaluate polynomially many inputs without losing its functionality. Accordingly, we define reusability below as a stronger version of these correctness notions.

\begin{definition}[Reusability] \label{def:reusable}
  Let $(\cp, \eval)$ be a $\delta$-correct copy-protection scheme for $\fclass$. Then, $(\cp, \eval)$ is called \textbf{reusable} if the following holds for every $m = \poly(\secparam)$, every $(x_1, \dots, x_{m}) \in X^{m}$ and every $f \in \fclass$: \begin{align*}
      \prob \bracketsSquare{ f(x_m) \from \eval(1^\secparam, \rho_f^{m}, x_m) \; : \; \substack{ \rho_f^{1} \from \cp(1^\secparam, d_f) \\ \rho_f^{i+1} \from \eval(1^\secparam, \rho_f^{i}, x_i), \; 1 \le i \le m-1 }} \ge \delta - \negl(\secparam).
  \end{align*}
  
  \par Similarly, let $(\cp, \eval)$ be a $(\distr_X, \delta)$-mean-correct copy-protection scheme for $\fclass$. Then, $(\cp, \eval)$ is called \textbf{reusable} if the following holds for any $m = \poly(\secparam)$ and any $f \in \fclass$:
  \begin{align*}
      \prob \bracketsSquare{ f(x_m) \from \eval(1^\secparam, \rho_f^{m}, x_m) \; : \; \substack{ \rho_f^{1} \from \cp(1^\secparam, d_f) \\ x_i \from \distr_X, \; 1 \le i \le m \\ \rho_f^{i+1} \from \eval(1^\secparam, \rho_f^{i}, x_i), \; 1 \le i \le m-1 }} \ge \delta - \negl(\secparam).
  \end{align*}
\end{definition}

\begin{remark} \label{rem:reusability}
In the plain model or the quantum-accessible oracle models, reusability is implied by correctness by \Cref{cor:aagan}. In particular, $\delta = 1 - \negl(\secparam)$ yields $\gamma = \negl(\secparam)$ and $\delta = 1 - 1 / \poly(\secparam)$ yields $\gamma = 1 / \poly(\secparam)$. However, in the classical-accessible setting such an implication is unclear due to the fact that the oracle queries force an intermediate measurement that cannot be pushed to the end.
\end{remark}

\paragraph{Security:} Security in the context of copy-protection means that given a copy-protected program $\rho_f$ of a function $f \in \fclass$, no QPT adversary can produce two programs that can both be used to compute $f$. This is captured in the following definition adapted from the "malicious-malicious security" definition in \cite{broadbent2021secure}: 
\begin{definition}[Piracy Experiment]
\label{def:piracyexperiment}
A \textbf{piracy experiment} is a game defined by a copy-protection scheme $(\cp, \eval)$, a distribution $\distr_\fclass$ over $\fclass$, and a class of distributions $\distrclass = \{\distrclass(f)\}_{f \in \fclass}$ over $X$. It is the following game between a challenger and an adversary, which is a triplet of algorithms $\abc$:
\begin{itemize}
    \item \textbf{Setup Phase:} The challenger samples a function $f\leftarrow \distr_\fclass$ and sends $\rho_f \leftarrow \cp(1^\secparam, d_f)$ to $\alice$.
    \item \textbf{Splitting Phase:} $\alice$ applies a CPTP map to split $\rho_f$ into a bipartite state $\rho_{BC}$; she sends the $B$ register to $\bob$ and the $C$ register to $\charlie$. No communication is allowed between $\bob$ and $\charlie$ after this phase.
    \item \textbf{Challenge Phase:} The challenger samples $(x_B, x_C) \leftarrow \distrclass(f) \times \distrclass(f)$ and sends $x_B, x_C$ to $\bob, \charlie$, respectively.
    \item \textbf{Output Phase:} $\bob$ and $\charlie$ output\footnote{Since $\bob$ and $\charlie$ cannot communicate, the order in which they use their share of the copy-protected program is insignificant.} $y_B \in Y$ and $y_C \in Y$, respectively, and send to the challenger. The challenger outputs 1 if $y_B = f(x_B)$ and $ y_C = f(x_C)$, indicating that the adversary has succeeded, and 0 otherwise.
\end{itemize}
The bit output by the challenger is denoted by $\pirateexp{\cp}{\eval}{\distr_\fclass}{\distrclass}(1^\secparam, \abc)$.

\end{definition}
\begin{definition}[Copy-Protection Security] \label{def:uncloneability}
Let $(\cp, \eval)$ be a copy-protection scheme for a class $\fclass$ of functions $f:X \to Y$. Let $\distr_\fclass$ be a distribution over $\fclass$ and $\distrclass = \left\{\distrclass(f)\right\}_{f \in \fclass}$ a class of distributions over $X$. Then, $(\cp, \eval)$ is called $\left(\distr_\fclass,\distrclass,\delta\right)$-\textbf{secure} if any QPT adversary $\abc$ satisfies \[\prob \left[ b = 1\ :\ b \leftarrow \pirateexp{\cp}{\eval}{\distr_\fclass}{\distrclass}\left(1^\secparam, \abc\right)  \right] \le \delta .\]

\end{definition}

\noindent Note that this definition is referred to as malicious-malicious security because the adversary is free to choose the registers $B,C$ as well as the evaluation algorithms used by $\bob$ and $\charlie$. \\

\subsection{Quantum-Unlearnability} \label{sec:prelims:unlearnability}
Below we state the definition of a quantum unlearnable circuit class from \cite{ALP20}, which states that a QPT adversary cannot output a quantum implementation of $C \in \cC$ given oracle access:
\begin{definition}
\label{def:unlearnability}
A circuit class $\cC$ is \textbf{$\nu$-quantum unlearnable} with respect to distribution $\distrc$ over $\cC$ if for any quantum adversary $\alice^C$ making at most $\poly(\secparam)$ queries, we have $$\prob\left[ \forall x, \prob[U^*(\rho^*,x)=C(x)]\geq \nu \ :\ \substack{C \leftarrow \distrc \\ (U^*,\rho^*) \leftarrow \cA^{C(\cdot)}(1^\secparam)} \right] \leq \negl(\secparam).$$ 
If there exists a distribution $\distr_C$ satisfying this property, then $\cktclass$ is simply called \textbf{$\nu$-quantum unlearnable}.
\end{definition}

\section{Impossibility in Classical-Accessible Random Oracle Model}
\noindent We show the infeasibility of copy-protection in the classical-accessible {\em random} oracle model (\carom). We prove this in two steps: 
\begin{itemize}
    \item First, we show the impossibility of approximately correct copy-protection (even without reusability) in the plain model. 
    \item Next, we show that any copy-protection in CAROM can be transformed into approximately-correct copy-protection in the plain model. Invoking the above result, it follows that copy-protection in CAROM is impossible. 
\end{itemize}

\newcommand{\newbfckt}{\mathsf{CC}}
\newcommand{\gcktclass}{\mathcal{G}}
\newcommand{\pqfhe}{\widetilde{\qfhe}}
\newcommand{\pkprime}{\widetilde{\mathsf{PK}}}
\newcommand{\skprime}{\widetilde{\mathsf{SK}}}

\subsection{Impossibility of Approximate Copy-Protection} \label{sec:impossibilityplain}
Following the approach of \cite{ALP20}, we will construct a class of circuits, which cannot be learned via oracle access, yet can be learned given the quantum circuit. The key tool enabling the latter is the power of homomorphic evaluation. The class of circuits we construct is related to but different from that constructed by \cite{ALP20}.

\begin{theorem}
\label{thm:imp:approxcp}
Assuming the existence of QFHE with malicious circuit privacy and QLWE, there exists an unlearnable class $\gcktclass$ of circuits of the form $G:  X \to Y$ and an input distribution $\distr_X$ over $X$ such that there exists no copy-protection scheme $(\cp, \eval)$ in the plain model for $\gcktclass$ with $(\distr_X, 1 - \varepsilon(\secparam))$-mean-correctness and $(\distr_\gcktclass, \distrclass, \delta(\secparam))$-security for any $(\varepsilon, \delta)$ satisfying $\delta \le 1 - 3\sqrt{\varepsilon}$ and for any $\distr_\gcktclass, \distrclass$.
\end{theorem}
\begin{proof}
Let $\lobf$ be a quantum-secure lockable obfuscation scheme (see \Cref{sec:prelims_lobf}) and $\qfhe$ be a quantum fully homomorphic encryption scheme satisfying malicious circuit privacy (see \Cref{sec:prelims_qfhe}). We first recall the circuit class $\cktclass$ used in the impossibility result of \cite{ALP20}, presented in~\Cref{fig:approxcp:ckt}. Every circuit in $\cktclass$ is of the form $C_{a,b,r,\pk,\cO}$, where $a,b,r,\pk,\cO$ are parameters described below. 
\begin{itemize}[noitemsep]
    \item $a,b,r \in \{0,1\}^{\secparam}$. 
    \item $(\pk,\sk)$ is in the support of $\qfhe.\gen(1^\secparam)$.
    \item $\cO$ is in the support of $\lobf(\newbfckt[\allowbreak \qfhe.\dec(\sk,\cdot),b,\allowbreak (\sk||r)])$.
\end{itemize}
\noindent

Let $\pqfhe=(\pqfhe.\setup, \pqfhe.\enc, \pqfhe.\dec, \pqfhe.\eval)$ be a QFHE scheme with malicious circuit privacy (\Cref{def:qfhe_circuitprivacy}). 
\par Using $\cktclass$, we define\footnote{In both $\cktclass$ and $\gcktclass$, we assume padding with zeros appropriately to make input/output lengths compatible.} another class of circuits $\gcktclass$ in~\Cref{fig:approxcpclass}. We show that $\gcktclass$ cannot be copy-protected. Every circuit in $\gcktclass$ is of the form $G_{C_{a,b,r,\pk,\cO}}$, which is a fixed circuit parameterized by a circuit $C_{a,b,r,\pk,\cO} \in \cktclass$. 

\begin{figure}[!htb]
\begin{center}
\begin{tabular}{|p{12cm}|}
\hline \\
\noindent {$\underline{C_{a,b,r,\pk,\cO}(x)}$:} \\
\begin{enumerate}
    \item If $x = 0$, output $\qfhe.\enc\left(\pk,a;r\right)|\cO|\pk$, where $\cO$ is generated as follows: \newline 
    $\allowbreak\cO \leftarrow \lobf(\newbfckt[\allowbreak \qfhe.\dec(\sk,\cdot),b,\allowbreak (\sk|r)])$, where $\newbfckt$ is a compute-and-compare circuit.
    \item Else if $x=a$, output $b$.
    \item Otherwise, output $0$
\end{enumerate}
 \\
     \hline
\end{tabular}
\end{center}
\caption{Description of a circuit $C_{a,b,r,\pk,\cO}$ in $\cktclass$.}
\label{fig:approxcp:ckt}
\end{figure}

\begin{figure}[!htb]
\begin{center}
\begin{tabular}{|p{12cm}|}
\hline \\
\noindent {$\underline{G_{C_{a,b,r,\pk,\cO}}(\pkprime,X)}$:} \\
\begin{enumerate}
    \item Parse $X$ as ciphertext $CT$. 
    \item Compute and output $CT^* \leftarrow \pqfhe.\eval_{\pkprime}(C_{a,b,r,\pk,\cO},CT)$.
\end{enumerate}
 \\
     \hline
\end{tabular}
\end{center}
\caption{Description of a circuit $G \in \gcktclass$, used for proving impossibility of approximate CP}
\label{fig:approxcpclass}
\end{figure}

\paragraph{Unlearnability.} We will first show that $\gcktclass$ is unlearnable with respect to a distribution $\distrc(\secparam)$, which we define as follows: $\distrc(\secparam)$ outputs a circuit from $\cktclass_{\secparam}$ by sampling $a,b,r \xleftarrow{\$} \{0,1\}^{\secparam}$, then computing $(\pk,\sk)\leftarrow \qfhe.\gen(1^\secparam)$, and finally computing an obfuscation $\cO \leftarrow \lobf(\newbfckt[\allowbreak \qfhe.\dec(\sk,\cdot), \allowbreak (b,\allowbreak (\sk|r))])$, where $\newbfckt$ is a compute-and-compare circuit. Since every $G_C \in \gcktclass$ is uniquely determined by the description of $C \in \cktclass$, $\distrc(\secparam)$ induces a distribution $\distr_\gcktclass(\secparam)$ on $\gcktclass$ given by $G_C: C \leftarrow \distrc$.

\par We cite the following result:
\begin{lemma}[Proposition 46 in \cite{ALP20}]
\label{lem:unlearnableC} For any non-negligible function $\nu = \nu(\secparam)$, $\cktclass$ is $\nu$-quantum unlearnable with respect to $\distrc(\secparam)$.
\end{lemma}
Using this as a black-box, we can now prove unlearnability of $\gcktclass$:
\begin{lemma}
\label{lem:unlearnableG}
$\gcktclass$ is unlearnable with respect to the distribution $\distr_\gcktclass(\secparam)$. 
\end{lemma}
\begin{proof}
Suppose that there exists an unbounded quantum adversary $\alice^{G_C(\cdot)}$ which makes at most $\poly(\secparam)$ queries and satisfies $$\prob\left[ \forall x, \prob[U^*(\rho^*,x)=G(x)]\geq \nu(\secparam) \ :\ \substack{C \leftarrow \distrc \\ (U^*,\rho^*) \leftarrow \cA^{G_C(\cdot)}(1^\secparam)} \right] \geq \mu(\secparam)$$
for some non-negligible functions $\nu, \mu$. We will use $\alice^{G_C(\cdot)}$ to construct an adversary which violates the unlearnability of $C$.\\
\noindent We will proceed in two steps:
\begin{itemize}
    \item By malicious circuit privacy of $\pqfhe$, there exists an unbounded quantum adversary $(\alice')^{C(\cdot)}$ which makes at most $\poly(\secparam)$ queries to $C(\cdot)$ and satisfies \begin{align}
    \label{eq:learning} &\prob\left[ \forall x, \prob[U^*(\rho^*,x)=G(x)]\geq \nu(\secparam) - \negl(\secparam) \ :\ \substack{C \leftarrow \distrc \\ (U^*,\rho^*) \leftarrow (\cA')^{C(\cdot)}(1^\secparam)} \right] \nonumber  \\
    &\geq \mu(\secparam) - \negl(\secparam).\end{align}
    $\alice'$ simply uses $(\pqfhe.\Sim, \pqfhe.\ext)$ and oracle access to $C(\cdot)$ to simulate the queries of $\alice.$
    \item Using $(\alice')^{C(\cdot)}$, we will construct $(\widetilde{\alice'})^{C(\cdot)}$ which makes at most $\poly(\secparam)$ queries to $C(\cdot)$ and satisfies \begin{align}
    \label{eq:learning2}&\prob\left[ \forall x, \prob[U^*(\rho^*,x)=C(x)]\geq \nu(\secparam) - \negl(\secparam) \ :\ \substack{C \leftarrow \distrc \\ (U^*,\rho^*) \leftarrow (\widetilde{\cA'})^{C(\cdot)}(1^\secparam)} \right] \nonumber  \\ &\geq \mu(\secparam) - \negl(\secparam).\end{align}
    $(\widetilde{\alice'})^{C(\cdot)}$ does the following: \begin{itemize}
        \item Run $(\alice')^{C(\cdot)}(1^\secparam)$ to obtain $(U, \rho)$
        \item Output $(U^*, \rho)$, where $U^*(\rho, x)$ does the following: \begin{enumerate}
            \item Compute $(\pk, \sk) \leftarrow \pqfhe.\setup(1^\secparam)$ and $\ct \leftarrow \pqfhe.\enc(\pk, x)$
                \item Compute $\ct^* \leftarrow U(\rho, (\pk, \ct))$.
            \item Output $y \leftarrow \pqfhe.\dec(\sk, \ct^*)$.
        \end{enumerate}
    \end{itemize}
    \ \\
    \noindent Conditioned on $(U, \rho)$ satisfying the event in \cref{eq:learning}, $(U, \rho^*)$ will satisfy the event in \cref{eq:learning2} by correctness of $\pqfhe$ evaluation, thereby violating $\mu$-quantum unlearnability of $\cktclass$.
\end{itemize}

\end{proof}

\begin{lemma}
There exists $\distr_X$ such that there is no copy-protection scheme $(\cp, \allowbreak \eval)$ in the plain model for $\gcktclass$ with $(\distr_X, 1 - \varepsilon(\secparam))$-mean-correctness and $(\distr_\gcktclass,\allowbreak \distrclass,\allowbreak \delta(\secparam))$-security for any $\distrclass$ and any $\varepsilon, \delta$ satisfying $\delta \le 1 - 3\sqrt{\varepsilon}$.
\end{lemma}
\begin{proof}
\noindent Define $$ \distr_{X} :=  \left\{  \pqfhe.\enc_{\pkprime}(x) : \substack{ (\pkprime, \skprime) \leftarrow \pqfhe.\setup(1^\secparam) \\ x \uniform \{0,1\}^{\poly(\secparam)}}  \right\} $$ 

Suppose there is an approximate copy-protection scheme $(\cp,\eval)$ for $\gcktclass$ with $1-\varepsilon(\secparam)$-mean-correctness. Let $\rho \leftarrow \cp(1^{\secparam},G_{C_{a,b,r,\pk,\cO}})$. We construct a QPT adversary $\adversary$ in~\Cref{fig:adversary} that given $\rho$, can recover an approximate classical description of $G$ with probability greater than $\delta(\secparam)$, hence violating security\footnote{The attack is simply to recover the classical description and send it to both parties at the splitting phase. Note that this attack succeeds irrespective of the test distributions $\distrclass$.}.\\

\begin{figure}[!htb]
\begin{center}
\begin{tabular}{|p{13cm}|}
\hline \\
\noindent {$\adversary(\rho)$:} \\
\begin{enumerate}
    \item Compute $(\pkprime, \skprime) \leftarrow \pqfhe.\setup(1^{\secparam})$. 
    \item Compute $CT \leftarrow \pqfhe.\enc_{\pkprime}(   0)$. 
     \item Run the copy-protection evaluation, $\rho' \otimes CT^* \leftarrow \eval(1^\secparam,\rho,CT)$. 
    \item Run the decryption, $(\ct_a|\cO|\pk) \leftarrow \pqfhe.\dec_{\skprime}(CT^*)$. Since $\ct_a$ is classical, maintain a copy of $\ct_a$. 
    \item Run the $\qfhe$ homomorphic evaluation, $$\ct_b^* \leftarrow \qfhe.\eval_{\pk}(\Phi,\ct_a),$$
    where $\Phi$ is a quantum circuit that on input $\sigma$, does the following:
    \begin{enumerate}
        \item Compute $\sigma_1\leftarrow\pqfhe.\enc(\pkprime,\sigma)$.
        \item \label{importantstep} Compute $\sigma_2 \leftarrow \qfhe.\eval( \pk,\eval(1^\secparam,\cdot,\cdot), \qfhe.\enc(\pkprime, \rho') ,\sigma_1)$.
        \item Compute $\sigma_3 \leftarrow \pqfhe.\dec(\skprime, \sigma_2)$ and output $\sigma_3$.
    \end{enumerate}
    \item Compute the unitary $\cO$ on $\ct_b^*$ and measure basis to obtain $(sk'|r')$. 
    \item Compute $a \leftarrow \qfhe.\dec(sk',\ct_a)$ and $b \leftarrow \qfhe.\dec(sk',\ct_b^*)$.
    \item Output $G_{C_{a,b,r',\pk,\cO}}$. 
\end{enumerate}
 \\
     \hline
\end{tabular}
\end{center}
\caption{Description of $\adversary$}
\label{fig:adversary}
\end{figure}

\textbf{Analysis of $\alice(\rho)$:}
By approximate correctness of $(\cp, \eval)$, we have \begin{align}
    &\prob \left[ G_C(x) \leftarrow \eval(1^\secparam, \rho,x) \ :\ \substack{\rho \leftarrow \cp(1^{\secparam},G_C)\\ \ {x \leftarrow \distr_X}} \right]\nonumber \\ &= \prob \left[ G_C(x) \leftarrow \eval(1^\secparam, \rho, \pqfhe.\enc_{\pkprime}(x)) \ :\ \substack{\rho \leftarrow \cp(1^{\secparam},G_C)\\ \ (\pkprime, \skprime) \leftarrow \pqfhe.\setup(1^\secparam)  \\
        \ x \uniform \{0,1\}^{\poly(\secparam)}} \right] \nonumber\\ \label{eq:meancorrectness}
        &\geq 1 - \varepsilon(\secparam).
\end{align}
Define 
\begin{align}
    \zeta_x := \prob \left[ G_C(x) \leftarrow \eval(1^\secparam, \rho, \pqfhe.\enc_{\pkprime}(x)) \ :\ \substack{\rho \leftarrow \cp(1^{\secparam},G_C)\\ \ (\pkprime, \skprime) \leftarrow \pqfhe.\setup(1^\secparam) } \right],
    \end{align}
    so that \cref{eq:meancorrectness} can be written as \begin{align}
    \label{eq:expectation}
        \E_{x \uniform \{0,1\}^{\poly(\secparam)}} \left[\zeta_x \right] \ge 1 - \varepsilon(\secparam).
\end{align}
We observe that $|\zeta_x - \zeta_{x'}| \le \negl(\secparam)$ for any $x,x' \in \{0,1\}^{\poly(\secparam)}$. To see this, suppose the difference is not negligible. Define an adversary $\alice'$ who can break semantic security of $\pqfhe$. Given a ciphertext $CT \leftarrow \pqfhe.\enc(\pkprime, x)$, $\alice'$ does the following: \begin{itemize}
    \item Sample $\rho \leftarrow \cp(1^{\secparam},G_C)$
    \item Output 1 if $G_C(x) \leftarrow \eval(1^\secparam, \rho, CT)$, and 0 otherwise.
\end{itemize}
$\alice'$ outputs 1 with probability $\zeta_x$, hence she can distinguish encryptions of $x$ and $x'$, contradiction. Hence, $|\zeta_x - \zeta_0| \le \negl(\secparam)$ for all $x$. Let $E_1$ be the event that $CT^* = \pqfhe.\enc(\pkprime, \qfhe.\enc_{\pk}(a)\allowbreak | \cO | \pk)$, then by \cref{eq:expectation} we have \begin{align} \label{eq:zeta0}
    \pr{E_1} = \zeta_0 \ge
    \E_{x \uniform \{0,1\}^{\poly(\secparam)}} \left[\zeta_x - \negl(\secparam) \right] \ge 1 - \varepsilon(\secparam) - \negl(\secparam).
\end{align}
At this stage, by modifying $\eval$ if necessary, we can assume $\trD{\rho'}{\rho} \le \sqrt{\varepsilon(\secparam)} + \negl(\secparam)$ by \Cref{cor:aagan}.
Let $E_2$ be the event $\alice$ succeeds in step 4; in particular, $E_2$ implies $\ct_a = \qfhe.\enc(\pk, a)$. By $\pqfhe$ correctness we have $\pr{E_2 \given E_1} \ge 1 - \negl(\secparam)$.

Conditioned on $E_2$, step \ref{importantstep} is nothing but a $\qfhe$ homomorphic evaluation of $\eval(1^\secparam, \rho',\sigma_1)$, where $\sigma_1 \leftarrow \pqfhe.\enc_{\pkprime}(\ct_a)$. \Cref{eq:zeta0} holds for any $\zeta_{x'}$ including $\zeta_{\ct_a}$, hence $\eval(1^\secparam, \rho, \sigma_1)$ will succeed with probability $1 - \varepsilon(x) - \negl(\secparam)$, i.e. it will output a $\qfhe$ encryption of $b$. Since $\rho'$ is $\sqrt{\varepsilon}$-close to $\rho$ in trace distance, it follows that $\eval(1^\secparam, \rho', \ct_a)$ will succeed with probability $1 - \varepsilon(\secparam) - \sqrt{\varepsilon(\secparam)} - \negl(\secparam)$. This inequality together with correctness of $\pqfhe$ and correctness of $\qfhe$ homomorphic evaluation imply that $\pr{E_3 \given E_2} \ge 1 - \varepsilon(\secparam) - \sqrt{\varepsilon(\secparam)} - \negl(\secparam) \ge 1 - 2\sqrt{\varepsilon(\secparam)}$,
where $E_3$ is the event that $\ct_b^* \in \qfhe.\enc_{\pk}(b)$.
Conditioned on $E_1, E_2, E_3$, the adversary can recover the true classical description with probability $1 - \negl(\secparam)$ by correctness of $\qfhe$ and $\lobf$. Thus, $\alice'$ succeeds in recovering the classical description of $G_C$ wiith probability at least \begin{align*} &\pr{E_3} - \negl(\secparam) \ge \pr{E_3 \given E_2} \cdot \pr{E_2 \given E_1} \cdot \pr{E_1} - \negl(\secparam) \ge \\
&\ge \left(1 - \varepsilon(\secparam) - \negl(\secparam)\right)(1 - \negl(\secparam))(1 - 2\sqrt{\varepsilon(\secparam)}) - \negl(\secparam) \\
&\ge 1 - 3\sqrt{\varepsilon(\secparam)},
\end{align*}
which suffices for the proof.

\end{proof}

\end{proof}
\newcommand{\aliceo}{{\alice}^\cO}
\newcommand{\bobo}{\bob^{\cO}}
\newcommand{\charlieo}{{\charlie}^\mathcal{O}}
\newcommand{\abco}{\left(\aliceo, \bobo, \charlieo\right)}
\newcommand{\alicetilde}{\widetilde{\alice}}
\newcommand{\bobtilde}{\widetilde{\bob}}
\newcommand{\charlietilde}{\widetilde{\charlie}}
\newcommand{\abctilde}{\left( \alicetilde, \bobtilde, \charlietilde \right)}
\newcommand{\outerproduct}[1]{\ketbra{#1}{#1}}

\subsection{Impossibility of Copy-Protection in \carom{}}
\label{sec:imp:carom}
Assuming impossibility of approximate copy-protection in the plain model, we will show that it is impossible to have a copy-protection scheme in \carom{} which satisfies $\brackets{1-\negl(\secparam)}$-correctness and $\brackets{1 - \allowbreak \frac{1}{\poly(\secparam)}}$-security. The proof will follow closely the proof in \cite{CKP15}, which shows impossibility of classical virtual-black-box obfuscation in the random oracle model assuming its impossibility in the plain model.

\begin{lemma}
\label{lem:romtoplain}
Let $\delta = \delta(\secparam)$ be a function and $\fclass$ be a class of functions of the form $f: X \to Y$. Let $\distr_\fclass, \distr_X$ be distributions over $\fclass, X$, respectively, and $\distrclass = \{\widetilde{\distr_X}^f\}_{f \in \fclass}$ be a class of distributions over $X$. Assume that $\distr_X$ can be efficiently sampled. Suppose there exists a reusable copy-protection scheme $(\cp^\mathcal{O}, \eval^\mathcal{O})$ in \carom{} for $\fclass$ with $\left(\distr_X, 1 - \negl(\secparam)\right)$-mean-correctness and $\left(\distr_\fclass, \distrclass, \delta\right)$-security. Then, for any noticeable function $\varepsilon$ there exists a copy-protection scheme $(\evalprime, \cpprime)$ in the plain model with $(\distr_X, 1 - \varepsilon(\secparam))$-mean-correctness and $\left(\distr_\fclass, \distrclass, \delta\right)$-security. 
\end{lemma}

\noindent Combining~\Cref{lem:romtoplain} with~\Cref{thm:imp:approxcp}, we have the following. 

\begin{theorem}[Main Theorem] \label{thm:imp_main}
There exists an unlearnable class of circuits $\gcktclass = \gcktclass(\secparam)$ and an input distribution $\distr_X$ such that there does not exist a reusable copy-protection scheme in \carom{} with $\left(\distr_X, 1 - \negl(\secparam)\right)$-mean-correctness and $(\distr_\fclass,\allowbreak \distrclass, \allowbreak 1 - \varepsilon(\secparam))$-security for any noticeable function $\varepsilon$ and any distributions $\distr_\fclass, \distrclass$.
\end{theorem}

\begin{proof}[Proof of \Cref{lem:romtoplain}]
Let\footnote{In \cite{CKP15}, the number of test executions ($T$) has an extra factor of $N$. Our modified analysis could also be applied to their construction to get rid of the factor of $N$.} $(\cp^\mathcal{O}, \eval^\mathcal{O})$ be given as in the lemma statement. Since $\cp^\cO$ and $\eval^\cO$ both run in polynomial time, the numbers of queries they make to $\cO$ are bounded by some $M = \poly(\secparam)$ and $N = \poly(\secparam)$, respectively. We will construct a valid copy-protection scheme $(\cpprime, \evalprime)$ for $\fclass$ in the plain model. Let $T := \lceil\frac{2M}{\varepsilon}\rceil$. \\ 

\noindent \textbf{Copy Protection:} $\cpprime(1^\secparam, d_f)$ takes as input the security parameter $1^\secparam$ and the classical description $d_f$ of a function $f \in \fclass$, and it does the following: \begin{enumerate}
    \item Simulate \textit{on-the-fly} a classical random oracle $\mathcal{O}$ to be used in any of the steps below.
    \item \label{step2} Run $\rho_f \leftarrow \cp^\mathcal{O}(1^\secparam, d_f)$.
    
    \item \label{step3} Set $\rho_f^1 := \rho_f$. Pick $S \uniform \{0,1,\dots,T-1\}$. For $i = 1$ to $S$: \begin{enumerate}
        \item Sample $x_i \leftarrow \distr_X$. 
        \item Compute $\rho_f^{i+1} \otimes \ket{y_i}\bra{y_i} \leftarrow \eval^\mathcal{O}(1^\secparam, \rho_f^i, x_i)$ and record the queries made by $\eval$ with their answers in a database $D_\eval^i$. 
    \end{enumerate}
    \item Sample random oracle answers $r_1, r_2, \dots, r_N \uniform Y$. These will be used by $\evalprime$ to answer queries not recorded in any of the databases $D_\eval^i$.
    \item Output the state $\widetilde{\rho_f}:= \rho_f^{S+1} \otimes \ket{D_\eval}\bra{D_\eval} \bigotimes_{j=1}^N \ket{r_j}\bra{r_j}$, where $D_\eval := \bigcup_{i=1}^S D_\eval^i$.
\end{enumerate}

\noindent \textbf{Evaluation:} $\evalprime(1^\secparam, \widetilde{\rho_f}, x)$ takes as input the security parameter $1^\secparam$, a copy-protected state $\widetilde{\rho}$, and an input $x \in X$, and it does the following: \begin{enumerate}
    \item Parse the state $\widetilde{\rho_f}$ as $\rho_f \otimes \ket{D}\bra{D} \bigotimes_{j=1}^N \ket{r_j}\bra{r_j}$ by measuring the registers corresponding to the database $D$ and the oracle answers $\{r_j\}_{j=1}^N$.
    \item Run $\rho_f' \otimes \ket{y}\bra{y} \leftarrow \eval^\mathcal{O}(1^\secparam, \rho_f, x)$, answering the oracle queries of $\eval$ as follows: to answer the $j$th query made by $\eval$, answer consistently if the query is in $D$, and answer using $r_j$ if the query is not in $D$ (without loss of generality $\eval$ does not make repeated queries). Output the state $\left( \rho_f' \otimes \ket{D}\bra{D} \bigotimes_{j=1}^N \ket{r_j}\bra{r_j} \right) \otimes \ket{y}\bra{y}$.
\end{enumerate}

\paragraph{Approximate Correctness:} Fix $f \in \fclass$. By $(\distr_X, 1 - \negl(\secparam))$-mean-correctness and reusability of $(\cp^\cO, \eval^\cO)$, we have \begin{align} \label{eq:reusable}
      \prob \bracketsSquare{ f(x_{S+1}) \from \eval^\cO(1^\secparam, \rho_f^{S+1}, x_{S+1}) \; : \; \substack{ \rho_f^{1} \from \cp(1^\secparam, d_f) \\ x_i \from \distr_X, \; 1 \le i \le S+1 \\ \rho_f^{i+1} \from \eval^\cO(1^\secparam, \rho_f^{i}, x_i), \; 1 \le i \le S }} \ge 1 - \negl(\secparam).
 \end{align}
 
 \newcommand{\event}{E}

\noindent Recall that $\evalprime$ emulates $\eval^\cO$, using the database $D = D_\eval$ to answer oracle queries and using the independent random answers $r_j$ when the query is not in $D$. This is equivalent to running $\eval^{\cO_D}(1^\secparam, \rho_f, x)$, where $\cO_D$ is defined as a random oracle conditioned to be consistent with $D$. Let $D_\cp$ be the set of queries made by $\cp$ in step \ref{step2} of $\cpprime$. Let $\event$ be the event that $\eval(1^\secparam, \rho_f, x)$ does not make any oracle query in $D_\cp \setminus D_\eval$ during the execution of $\evalprime$.\footnote{Here by a slight abuse of notation we only consider the input $x \in X$ of the query and not the answer $y \in Y$.} Keep in mind that $\evalprime$ emulates $\eval^\cO$ at the time of the $(S+1)$st test execution. Hence, $\event$ is equivalently the event that $\eval(1^\secparam, \rho_f^{S+1}, x_{S+1})$ makes no query in $D_\cp \setminus D_\eval$, where $\rho_f^{S+1}$ and $x_{S+1}$ are as in \cref{eq:reusable}. Conditioned on $\event$, the emulation of $\evalprime$ is flawless, i.e. $\eval^\cO(1^\secparam, \rho_f, x)$ and $\eval^{\cO_D}(1^\secparam, \rho_f, x)$ are perfectly indistinguishable. Hence, we can lower-bound the correctness of $\evalprime$ by lower-bounding the probability of $\event$:
\begin{align}
 &\pr{ f(x) \leftarrow \evalprime(1^\secparam, \widetilde{\rho_f},x) \ :\ \substack{\widetilde{\rho_f} \leftarrow \cpprime(1^{\secparam},d_f)\\ \ {x \leftarrow \distr_X}} } \nonumber \\
 &= 1 - \pr{ y \ne f(x) \ :\ \substack{\widetilde{\rho_f} \leftarrow \cpprime(1^{\secparam},d_f)\\ \ {x \leftarrow \distr_X} \\ y \leftarrow \evalprime(1^\secparam, \widetilde{\rho_f}, x)} } \nonumber \\
 &= 1 - \pr{ y \ne f(x_{S+1}) \ :\ \substack{\rho_f \leftarrow \cp^\cO(1^{\secparam},d_f)\\
 S \uniform \{0,1,\dots,T-1\} \\
 x_i \from \distr_X, \; 1 \le i \le S+1 \\ \rho_f^{i+1} \from \eval^\cO(1^\secparam, \rho_f^{i}, x_i), \; 1 \le i \le S \\ y \leftarrow \eval^{\cO_D}(1^\secparam, \rho_f^{S+1},x_{S+1})} } \nonumber \\
 &\ge 1 - \left( \pr{\event} \cdot \pr{ y \ne f(x_{S+1}) \ :\ \substack{\rho_f \leftarrow \cp^\cO(1^{\secparam},d_f)\\
 S \uniform \{0,1,\dots,T-1\} \\
 x_i \from \distr_X, \; 1 \le i \le S+1 \\ \rho_f^{i+1} \from \eval^\cO(1^\secparam, \rho_f^{i}, x_i), \; 1 \le i \le S \\ y \leftarrow \eval^{\cO_D}(1^\secparam, \rho_f^{S+1},x_{S+1})} \; \bigg\rvert \; \event} + \pr{\neg \event} \right) \nonumber \\ \label{eq:lowerbound}
 &\overset{(\ref{eq:reusable})}{\ge} 1 - \pr{\neg \event} - \negl(\secparam).
\end{align}
\noindent Therefore, the following claim will suffice for the proof together with \cref{eq:lowerbound}:
\begin{claim}
\label{clm:eventbound}
$ \pr{\neg \event} \le \frac{\varepsilon(\secparam)}{2}. $
\end{claim}
\begin{proof} Consider the following experiment, which consists of $T$ executions of $\eval^\cO$ with random inputs: \begin{enumerate}
    \item Let $\cO$ be a random oracle. 
    \item Compute $\rho_f \leftarrow \cp^\mathcal{O}(1^\secparam, d_f)$.
    
    \item Set $\rho_f^1 := \rho_f$. For $i = 1$ to $T$: \begin{enumerate}
        \item Sample $x_i \leftarrow \distr_X$.
        \item Compute $\rho_f^{i+1} \otimes \ket{y_i}\bra{y_i} \leftarrow \eval^\mathcal{O}(1^\secparam, \rho_f^i, x_i)$. Let $D_\eval^i$ be the set of inputs queried by $\eval$ in this step. 
    \end{enumerate}
\end{enumerate}

\noindent Let $w_i = \left| \brackets{D_\eval^i \cap D_\cp} \setminus \bigcup_{1 \le j < i} D_\eval^j\right|$ be the random variable corresponding to the number of \emph{new} queries from $D_\cp$ made by $\eval^\cO$ at the $i$th step above. Let $p_i = \pr{w_i \ge 1}$ be the probability that a new query is made at the $i$th step. Recall that $\evalprime$ chooses the number of test executions uniformly at random, hence $\pr{\neg \event} = \frac{1}{T} \sum_{i=1}^T p_i$ is the probability that $w_i \ge 1$ for a random $i \in [T]$. Now, by linearity of expectation we have \begin{align*}
    \sum_{i=1}^T p_i \le \sum_{i=1}^T \E\bracketsSquare{w_i} =  \E\bracketsSquare{\sum_{i=1}^T w_i } = \E\bracketsSquare{\bigg| D_\cp \cap \brackets{\bigcup_{j=1}^T D_\eval^j} \bigg|} \le |D_\cp| \le M.
\end{align*} Therefore, $\pr{\neg E} = \frac{1}{T} \sum_{i=1}^T p_i \le \frac{M}{T} \le \frac{\varepsilon}{2}$ as desired.
\end{proof}

\paragraph{Copy-Protection Security:} Suppose there is an adversary $\abctilde$ which succeeds in the pirating experiment for $(\cpprime, \evalprime)$ with probability $\varepsilon$, i.e. \[ \prob \left[ b = 1\ :\ b \leftarrow \pirateexp{\cpprime}{\evalprime}{\distr_\fclass}{\distrclass}\left(1^\secparam, \abctilde\right)  \right] = \varepsilon. \] Using $\abctilde$, we will construct an adversary $\abco$ which succeeds in the pirating experiment for $(\cp, \eval)$ with probability $\varepsilon$. This will immediately imply $\delta$-security of $(\cpprime, \evalprime)$ by the $\delta$-security of $(\cp, \eval)$. \\

\noindent  We set $\bobo$ and $\charlieo$ to be identical to $\bobtilde$ and $\charlietilde$, respectively, so that they do not make any queries. We define $\aliceo$, which has oracle access to $\mathcal{O}$, as follows: \begin{enumerate}
    \item Given a copy-protected program $\rho_f =: \rho_f^1$, pick $S \uniform \{0,1,\dots,T-1\}$ and repeat for $i=1$ to $S$: \begin{enumerate}
        \item Sample $x_i \leftarrow \distr_X$
        \item Run $\rho_f^{i+1} \otimes \ket{y_i}\bra{y_i} \leftarrow \eval^\mathcal{O}(1^\secparam, \rho_f^i, x_i)$, forwarding the oracle queries of $\eval$ to $\mathcal{O}$ and recording them with their answers in a database $D_\eval^i$.
    \end{enumerate}
    \item Sample random oracle answers $r_1, r_2, \dots, r_N \uniform Y$.
    \item Set $\widetilde{\rho_f}:= \rho_f^{S+1} \otimes \ket{D_\eval}\bra{D_\eval} \bigotimes_{j=1}^N \ket{r_j}\bra{r_j}$, where $D_\eval := \bigcup_{i=1}^S D_\eval^i$.
    \item Run $\alicetilde$ on the state $\widetilde{\rho_f}$.
\end{enumerate}

By construction, the bipartite state received by $\bobo$ and $\charlieo$ is identical to that received by $\bob$ and $\charlie$. This is because the on-the-fly simulation done by $\cpprime$ is perfectly indistinguishable from the real oracle answers provided by $\mathcal{O}$. Therefore, \begin{align*} &\prob \left[ b = 1\ :\ b \leftarrow \pirateexp{\cp}{\eval}{\distr_\fclass}{\distrclass}\left(1^\secparam, \abc\right)  \right] \\ = &\prob \left[ b = 1\ :\ b \leftarrow \pirateexp{\cpprime}{\evalprime}{\distr_\fclass}{\distrclass}\left(1^\secparam, \abctilde\right)  \right] = \varepsilon \end{align*} as desired.

\end{proof}

\submversion{\bibliographystyle{plain}}
\fullversion{\bibliographystyle{alpha}}
\bibliography{crypto}

\newpage

\end{document}